%% file: parallel_block_tree.tex
\documentclass[a4paper,UKenglish,cleveref, autoref, thm-restate]{lipics-v2021}

\pdfoutput=1 %
\hideLIPIcs  %

\usepackage{microtype}

\usepackage{amsmath}

\usepackage{booktabs}

\usepackage{tikz}
\usetikzlibrary{trees}
\usepackage{pgfplots}
\pgfplotsset{compat=1.18}

\usepgfplotslibrary{groupplots}

\pgfplotsset{
  major grid style={thin,dotted},
  minor grid style={thin,dotted},
  xmajorgrids,
  ymajorgrids,
  yminorgrids,
  axis lines*=left,
  xlabel near ticks,
  ylabel near ticks,
  axis lines*=left,
  every axis/.append style={
    line width=0.7pt,
    tick style={
      line cap=round,
      thin,
      major tick length=4pt,
      minor tick length=2pt,
    },
  },
  legend style={
    line width=0.25pt,
    /tikz/every even column/.append style={column sep=3mm,black},
    /tikz/every odd column/.append style={black},
  },
}

\definecolor{my-dark-red}{RGB}{183, 28, 28}
\definecolor{my-red}{RGB}{244,67,54}
\definecolor{my-pink}{RGB}{233,30,99}
\definecolor{my-purple}{RGB}{156,39,176}
\definecolor{my-deep-purple}{RGB}{103,58,183}
\definecolor{my-indigo}{RGB}{63,81,181}
\definecolor{my-blue}{RGB}{33,150,243}
\definecolor{my-light-blue}{RGB}{3,169,244}
\definecolor{my-cyan}{RGB}{0,188,212}
\definecolor{my-teal}{RGB}{0,150,136}
\definecolor{my-green}{RGB}{76,175,80}
\definecolor{my-light-green}{RGB}{139,195,74}
\definecolor{my-lime}{RGB}{205,220,57}
\definecolor{my-yellow}{RGB}{255,235,59}
\definecolor{my-amber}{RGB}{255,193,7}
\definecolor{my-orange}{RGB}{255,152,0}
\definecolor{my-deep-orange}{RGB}{255,87,34}
\definecolor{my-brown}{RGB}{121,85,72}
\definecolor{my-grey}{RGB}{158,158,158}
\definecolor{my-blue-grey}{RGB}{96,125,139}
\definecolor{my-lipics-grey}{rgb}{0.6,0.6,0.61}

\pgfplotscreateplotcyclelist{colorList}{
  {my-indigo,mark=triangle,densely dashed,mark options={solid}},
  {my-brown,mark=square},
  {my-deep-purple,mark=diamond,densely dashed,mark options={solid}},
  {my-deep-orange,mark=pentagon,,densely dashed,mark options={solid}},
  {my-light-green,mark=*},
}

\title{Practical Parallel Block Tree Construction: First Results} %

\author{Robert Clausecker}{Zuse Institute Berlin, Germany}{clausecker@zib.de}{}{}
\author{Florian Kurpicz}{University of Münster, Germany}{florian.kurpicz@uni-muenster.der}{https://orcid.org/0000-0002-2379-9455}{}
\author{Etienne Palanga}{TU Dortmund University, Germany}{etienne.palanga@tu-dortmund.de}{}{}

\authorrunning{R. Clausecker, F. Kurpicz, and E. Palanga} %



\category{} %

\relatedversion{} %

\supplementdetails[subcategory={Source code}]{Software}{https://github.com/pasta-toolbox/block_tree}

\nolinenumbers %

\EventEditors{John Q. Open and Joan R. Access}
\EventNoEds{2}
\EventLongTitle{42nd Conference on Very Important Topics (CVIT 2016)}
\EventShortTitle{CVIT 2016}
\EventAcronym{CVIT}
\EventYear{2016}
\EventDate{December 24--27, 2016}
\EventLocation{Little Whinging, United Kingdom}
\EventLogo{}
\SeriesVolume{42}
\ArticleNo{23}

\begin{document}

\maketitle

\begin{abstract} \small\baselineskip=9pt
  The block tree [Belazzougui~et~al., J.~Comput.~Syst.~Sci.\,'21] is a compressed representation of a length-\(n\) text that supports access, rank, and select queries while requiring only \(O(z\log\frac{n}{z})\) words of space, where \(z\) is the number of Lempel-Ziv factors of the text.
  In other words, its space-requirements are asymptotically similar to those of the compressed text.
  In practice, block trees offer comparable query performance to state-of-the-art compressed rank and select indices.
  However, their construction is significantly slower.
  Additionally, the fastest construction algorithms require a significant amount of working memory.
  To address this issue, we propose fast and lightweight parallel algorithms for the efficient construction of block trees.
  Our algorithm achieves similar speed than the currently fastest construction algorithm on one core and is up to four times faster using 64 cores.
  It achieves all that while requiring an order of magnitude less memory.
  As result of independent interest, we present a data parallel algorithm for Karp-Rabin fingerprint computation.
\end{abstract}

\newcommand{\access}[1]{\ensuremath{\mathit{access}(#1)}}
\newcommand{\rank}[2]{\ensuremath{\mathit{rank}_{#1}(#2)}}
\newcommand{\select}[2]{\ensuremath{\mathit{select}_{#1}(#2)}}
\newcommand{\ceil}[1]{\ensuremath{\lceil#1\rceil}}
\newcommand{\floor}[1]{\ensuremath{\lfloor#1\rfloor}}

\section{Introduction}
In today's information age, textual data form of DNA and protein sequences, source code, and textual content like news and Wikipedia articles is created at a much higher rate than hardware is advancing.
Therefore, we cannot handle the amount of data by just investing in a more potent processing infrastructure anymore.
Instead, we require efficient---scalable---data structures.
Here, ``scalable'' refers to (1) an efficient parallel processing of the data and (2) the compression of the data structure.
Both dimensions help to better utilize modern hardware by allowing us to process more data in the same time and available space.

For compression, we can exploit that many inputs are highly repetitive.
For example, if we consider DNA, humans shares around 99,9\,\% of their DNA with each other, making DNA sequences highly repetitive.
The same holds for versionized inputs like code repositories and versionized texts like articles in the Wikipedia.

Some of the most fundamental queries we can answer on a text \(T\) of length \(n\) over an alphabet of size \(\sigma\) are access, rank, and select queries:
\begin{itemize}
\item \(\access{i}=T[i]\) is the character at the \(i\)-th position of the text,
\item \(\rank{\alpha}{i}=|{j\leq i \colon T[j]=\alpha}|\) is the number of occurrences of \(\alpha\) in the prefix \(T[1..i]\)), and
\item \(\select{\alpha}{i}=\arg\min_{j}(\rank{\alpha}{j}=i)\) is the position of the \(i\)-th occurrence of \(\alpha\) in \(T\).
\end{itemize}
Applications include among others compressed full-text indices (for pattern matching)~\cite{Ferragina04fmindex,Navarro07index,GogP2014Optimized,Claude12grammarindex,Claude16universal}, lossless data compression~\cite{FerraginaGM2009MyriadWT,GrossiGV2003WaveletTree,Koppl22holz}, and computational geometry~\cite{Chien15gbwt}.

Currently, \emph{wavelet trees}~\cite{GrossiGV2003WaveletTree} are the ``defacto standard''~\cite{AlankoBPV2023SubsetWT} when it comes to answering these queries in compressed space.
However, the main flaw of wavelet trees is that they can only be compressed statistically.
This leaves the compression ratio on highly compressible inputs to be desired, as we always require at least one bit per character.
\emph{Block trees}~\cite{BelazzouguiCGGK2021BlockTrees} solve this problem by utilizing dictionary compression.
Refer to \cref{sec:preliminaries} for an overview and comparison of compression techniques.

We discuss related rank and select data structures in \cref{sec:related_work}.
In \cref{sec:block_trees}, we give an overview of block trees including a detailed look at their construction.
We present and analyze our new parallel construction algorithms in \cref{sec:parallel_block_tree_construction}.
We then provide an experimental evaluation of our implementation in \cref{sec:experimental_evaluation}.
Finally, in \cref{sec:conclusion_and_future_work}, we summarize the results and discuss possible future work and open problems.

\subparagraph*{Our Contributions.}
Block tree construction has not yet been considered in a parallel setting.
To the best of our knowledge, all previous work~\cite{BelazzouguiCGGK2021BlockTrees,KopplKM2023LPFBlockTrees} considers the sequential and external memory case.
While the current state-of-the-art (sequential) construction algorithm~\cite{KopplKM2023LPFBlockTrees} supports parallel construction, it is more of a byproduct of the utilized data structures that can be constructed in parallel.
Since the construction algorithm itself has not been parallelized, the scalability is severely lacking, cf.~our experimental evaluation in \cref{sec:experimental_evaluation}.
Furthermore, the algorithm's huge memory requirements---both in theory and in practice---limit its applicability on large inputs.

Our contribution is the \emph{first} dedicated parallel block tree construction algorithm that provides a trade-off between scalability and memory requirements, see \cref{sec:parallel_block_tree_construction}.
To this end, we employ domain decomposition, where we initially partition the input, then run our construction algorithm in parallel on the partitions, and merge the results throughout the computation.
By controlling the number of partitions, we obtain our trade-off.
In our experimental evaluation, in \cref{sec:experimental_evaluation}, we see that the space-overhead during construction is up to an order of magnitude lower than previous block tree construction algorithms.
Also, the COST~\cite{McSherryIM2015Cost}, i.e., the number of threads needed such that our algorithm is faster than the sequential state-of-the-art is only two.
As as result of independent interest, in \cref{sec:simd_karp_rabin}, we also show how to compute Karp-Rabin fingerprints~\cite{KarpR1987Fingerprints} using SIMD.

\section{Preliminaries}
\label{sec:preliminaries}
A text \(T\in[1,\sigma]^n\) is a text of length \(n\) over an alphabet of size \(\sigma\).
\(T[i,j]\) is the substring \(T[i]T[i+1]\dots T[j]\) for \(1\leq i\leq j\leq n\).
We also use \(|T[i,j]|=j-i+1\) to denote the length of a sub(string).
We can compare any two substrings in constant time with high probability using a rolling hash function for strings called Karp-Rabin fingerprint~\cite{KarpR1987Fingerprints}.
Let \(T\in[1..\sigma]^n\) be a text, \(q\in\Theta(n^c)\) for some constant \(c>1\) and \(r<q\) with \(r\nmid q\) be positive integers.
Then, the \emph{Karp-Rabin fingerprint} of \(T[s..e]\) is \(\phi(s,e)=\left(\sum_{i=s}^eT[i]\cdot r^{e-i}\right)\textnormal{mod}~q\).

\newcommand{\probP}{\textnormal{Prob}}

Now, for every two substring \(T[i..i+\ell]\) and \(T[j..j+\ell]\), we get that \(\phi(i,i+\ell)=\phi(j,j+\ell)\) if \(T[i..i+\ell]=T[j..j+\ell]\).
Otherwise, if \(T[i..i+\ell]\neq T[j..j+\ell]\) the probability of the corresponding fingerprints matching is \(\probP[\phi(i,i+\ell)=\phi(j,j+\ell)]\in O(\frac{\ell\log\sigma}{n^c})\).
Karp-Rabin fingerprints are also \emph{rolling} hash function.
This allows us, among others, to compute all length-\(\ell\) fingerprints of a text \(T\in[1,\sigma]^n\) in \(O(n)\) time.

\subsection{Measures of Compressability}
There exists a zoo of different measures of compressability.
We refer to the excellent survey by Navarro~\cite{Navarro2021RepetitiveMeasures} for a detailed overview.
The \emph{empirical entropy}~\cite{KosarajuM1999Entropy} is based on the distribution of \(T\)'s characters (or substrings).
The \(0\)-th order entropy is \(H_o(T)=\sum_{\alpha\in\Sigma}\frac{n_\alpha}{n}\log\frac{n}{n_\alpha}\) with \(n_\alpha=\rank{\alpha}{n+1}\).
For \(k>0\), we have \(H_k(T)=\sum_{s\in\Sigma^k}\frac{|T_s|}{n}H_0(T_s)\), where \(T_s\) is the subsequence of characters occurring directly after \(s\in[1,\sigma]^k\) in \(T\), e.g., for \(T=abcabd\) we have \(T_{ab}=cd\).
Overall, we get \(\log\sigma\geq H_0(T)\geq H_1(T)\geq\dots\geq H_n(T)\).
The big disadvantage of such a statistical compression is that it does not capture repetitions, as it considers the text on a character-by-character basis.

Another type of measure is the size of a \emph{Lempel-Ziv} factorization \(z\) of a text.
Such a factor can be stored as reference to an earlier occurrence in the text.
One of the most prominent techniques here are Lempel-Ziv-based compression algorithms.
The Lempel-Ziv 77 (LZ77) factorization~\cite{ZivL1977LZ77} parses a text \(T\) into factors \(f_1,\dots,f_z\in\Sigma^+\) such that \(T=f_1\dots f_z\) and for all factors \(i\in[1,z]\) there is a position \(j\in[1,\dots,|f_1\dots f_i|\) in the parsed text (including the new factor) such that \(f_i=T[j\dots j+|f_i|]\).

Strings can be represented as context-free \emph{grammar}~\cite{KiefferY2000GrammarCompression}.
The size of the smallest grammar \(g\) that only generates the text is NP-complete.
However, grammars of size \(O(z\log\frac{n}{z})\) can be computed in linear time~\cite{Rytter2003GrammarInZ}.
The currently best known measure is the \emph{string complexity} \(\delta\)~\cite{KociumakaNP2020Delta}.
It is \(\delta=\max_{k\in[1,n]}\frac{d_k}{k}\) with \(d_k=|\{T[i..i+k-1]\colon i\in[0,n-k]\}|\).

\subsection{Model of Computation}
We analyze our algorithms in the PRAM model.
Here we have multiple processors (PEs, processing elements) that share their memory.
The PRAM model is synchronized, i.e., in each time step, all PEs execute exactly one instruction.
It comes in multiple flavors that distinguish between exclusive and concurrent access to memory cells during one time step.
We further differentiate between read and write access.

The weakest model is the EREW PRAM where only exclusive read and exclusive write access to memory cells is allowed.
Slightly stronger is the CREW PRAM model.
Here, PEs are allowed to concurrently read from the same memory cell.
Still, only one PE is allowed to write to the same memory cell (during the same time step).
The strongest model we are considering allows concurrent read and write access (CRCW).
We consider two variants of the CRCW model: The Common-CRCW model allows multiple PEs to write to the same memory cell at the same time only if all of them write the same value.
The Arbitrary-CRCW allows allows multiple PEs to write tot the same memory cell without any restrictions on the value.
However, there is no way to determine which of the written values is stored.%

When comparing algorithms in the PRAM model, we are interested in their \emph{time} and \emph{work}.
The time is the number of time steps of a PRAM algorithm and the work is the total number of all operations (arithmetic on local data and reading and writing memory cells).
The work is the same as the running time of the algorithm when executed with a single PE.

The difference between the weaker (EWER) and stronger (CRCW) models become visible in the time required for simple building blocks.

\newcommand{\sbigotimes}{%
 \mathop{\mathchoice{\textstyle\bigotimes}{\bigotimes}{\bigotimes}{\bigotimes}}%
}
\begin{lemma}[All-Prefix Operation \cite{GoldbergZ1995PrefixSums}]
  \label{lem:all_prefix}
  Given \(n\) integers \(a_1,\dots,a_n\) and a binary associative operator \(\otimes\) that requires \(O(1)\) time.
  Then, in the EWEW model, the sequence \((s_i,s_2,\dots,s_n)\) with \(s_i=\sbigotimes_{j=1}^i a_j\) can be computed in \(O(\log n)\) time, \(O(n)\) work, and \(O(n)\) space.
\end{lemma}

\begin{lemma}[\cite{ColeV1989FasterParallelPrefixSum}]
  \label{lem:prefix_sum}
  In the Common-CRCW model, the All-Prefix operation (\cref{lem:all_prefix}) with \(\otimes=+\) (addition) can be computed in \(O(\log n/\log\log n)\) time, \(O(n)\) work, and \(O(n)\) space.
\end{lemma}

In addition to the All-Prefix operation, we often need to sort data in parallel.

\begin{lemma}[\cite{Cole1988EREWMergeSort}]
  \label{lem:sorting_ew}
  In the EREW and CREW PRAM models, sorting \(n\) integers requires \(O(\log n)\) time, \(O(n\log n)\) work, and \(O(n)\) words of space.
\end{lemma}

\begin{lemma}[\cite{Cole1988EREWMergeSort}]
  \label{lem:sorting_cw}
  In the CRCW PRAM model, sorting \(n\) integers requires \(O(\frac{\log n}{\log\log\log n})\) time using \(n\log^c n\) PEs for some constant \(c\geq 1\). 
\end{lemma}

\section{Related Work}
\label{sec:related_work}
There exists a plethora of compressed data structures that can answer access, rank, and select queries, however, with the wavelet tree as only exception, there exists very little work regarding parallelization.\footnote{For full-text indices with pattern matching support, there exists a large body of work on parallel construction. Since this is out of scope for this paper, we refer to a survey~\cite{BingmannD0KO02022SPPBook}.}

\subparagraph*{Bit Vectors.}
In the simplest setting---when we consider a binary alphabet of size two---these queries can be answered in constant time requiring only sublinear space~\cite{ClarkM1996Select,Jacobson1989LOUDS}.
There exist many data structures utilizing different techniques to achieve succinct~\cite{PibiriK2021MutableRankSelect,NavarroP2012CombinedSampling,Vigna2008BroadwordRankSelect,ZhouAK2013PopcountRankSelect,MarchiniV2020DynamicRankSelect,GonzalezGMN2005PracticalRankSelect,Kurpicz2022RankSelect} and compressed~\cite{RamanRR2007rrr,Okanohara2007Practical,Patrascu2008Succincter,GogP2014Optimized,BoffaFV2021LearnedBitVector} bit vectors with rank and select support.
To the best of our knowledge, no work on parallel (construction) for rank and select support on bit vectors (outside of word parallelism) has been published.

\subparagraph{Wavelet Trees.}
Wavelet trees~\cite{GrossiGV2003WaveletTree} generalize rank and select support to larger alphabets.
In simple terms, for a text \(T\in[1,\sigma]^n\), a wavelet tree consists of \(\ceil{\log\sigma}\) length-\(n\) bit vectors with (binary) rank and select support.
These bit vectors contain bits of the subsequences of the original text and are used to answer the queries.
For more details on the structure of wavelet trees and their many applications, we refer to the multiple surveys written on them~\cite{Navarro2014WaveletForAll,Makris2012WaveletSurvey,GrossiVX2011WaveletSurvey,FerraginaGM2009MyriadWT,DinklageEFKL2021PracticalWaveletTrees}.
Wavelet trees can also be statistically compressed, resulting in a data structure that requires \(n\ceil{H_0(T)}(1+o(1))\) bits of space and can answer rank and select queries in \(O(\log\sigma)\) time~\cite{GrossiGV2003WaveletTree}.
The query time can be improved to \(O(1+\frac{\log\sigma}{\log\log n})\) using multi-ary wavelet trees~\cite{FerraginaMMN2007MultiAryWaveletTree}.
Due to the great importance of wavelet trees, their construction has been well-researched both in theory~\cite{BabenkoGKS2015WT,MunroNV2016WT} and in practice~\cite{Kaneta2018VectorizedWXConstruction,DinklageFKT2023SIMDWT,EllertK2019ExternalWX} with results in shared~\cite{LabeitSB2017ParallelWXSACA,SepulvedaEFS2017DomainDecomposition,FischerKL2018PWX} and distributed~\cite{DinklageFK2020DistributedWX} memory.
Due to their importance, there exists also research improving query performance in practice~\cite{ClaudeNP2015WaveletMatrix,CereginiKV2025PredictiveModelWT,HongBGLZ2024WaveletForests}.

\subparagraph*{Grammar Compressed and Further Rank and Select Data Structures.}
In addition to wavelet trees, there exist statistically compressed data structures that require \(nH_k(T)+o(n\log\sigma)\) bits of space and can answer rank queries \(O(\log\log_w\sigma)\) time, select queries in \(\omega(1)\) time, and access in \(O(1)\) time~\cite{belazzougui15wavelet}.
To tackle the disadvantage of statistical compression on highly repetitive inputs, grammar compression has also been considered.
Using \(O(g\sigma)\) words of space, the grammar indices can answer rank and select queries in \(O(\log n)\) time~\cite{belazzougui15grammarrank,pereira17grammar}.
The parallel construction of grammars (without query support) has also been considered~\cite{MatsushitaI2021ParallelGrammar,MatsushitaI2022ParallelGrammar}.

\section{Block Trees}
\label{sec:block_trees}
The big disadvantage of statistically compressed data structures is that they require at least one bit per character of the input.
When considering highly repetitive inputs, this wastes a lot of space.
To tackle this problem, we consider block trees.
Block trees can answer access, rank, and select queries, making them an ideal drop-in replacement for wavelet trees.
They also  have successfully been applied to the compression of the suffix tree topology as well as the suffix array and its inverse~\cite{caceres22faster}.
Furthermore, adaptations of the block tree can be used to simulate \(k^2\)-trees~\cite{brisaboa18blocktree} and for locating patterns~\cite{navarro17blocktree}.

A block tree~\cite{BelazzouguiCGGK2021BlockTrees} is a tree where the root has out-degree \(s\) and all other inner nodes have out degree \(\tau\).
For simplicity, we assume that the text we build the block tree for is \(T\in[1,\sigma]^n\) with \(n=s\cdot\tau^h\) for some integer \(h\).
Here, \(h\) is also the height of the block tree.

Now, that we are aware of the shape of a block tree, we look at the content of the nodes.
Each node \(i\) in the block tree represents a substring of \(T\) called \emph{block} \(B_i\).
We start with the root, which represents the whole text.
Each of the root's children represents a consecutive length-\((n/s)\) block.
On each level of the block tree, we call two blocks \emph{consecutive} if they are consecutive in \(T\).
Let \(\alpha=B_{i}\cdot B_{i+1}\) be the concatenation of two consecutive blocks.
We mark blocks \(B_i\) and \(B_{i+1}\) if this is the leftmost occurrence of \(\alpha\) on this level.

All marked blocks are inner nodes with \(\tau\) children, which represent consecutive blocks on the next level.
An unmarked block \(B_u\), on the other hand, is a leaf that only store a pointer towards the pair of consecutive (marked) blocks \(B_i\cdot B_{i+1}\) that contain the leftmost occurrence of \(B_u\) and its offset in the two blocks.
See \cref{fig:example_block_tree} for an example.

\begin{figure}
  \centering
  \begin{tikzpicture}[
  grow=down,  %
  sibling distance=5em, %
  level distance=4em, %
  edge from parent/.style={draw, -latex},  %
  every node/.style={font=\footnotesize\ttfamily, draw, rectangle}, %
  level 1/.style={sibling distance=12em}, %
  level 2/.style={sibling distance=4em}, %
  level 3/.style={sibling distance=2em}  %
]

\node{abrainadrain\vphantom{jd}}
  child {node {abrain\vphantom{jd}}
    child {node {ab\vphantom{jd}}
      child {node {a\vphantom{jd}}}
      child {node {b\vphantom{jd}}}
    }
    child {node (t1) {ra\vphantom{jd}}
      child {node {r\vphantom{jd}}}
      child {node {a\vphantom{jd}}}
    }
    child {node (t2) {in\vphantom{jd}}
      child {node {i\vphantom{jd}}}
      child {node {n\vphantom{jd}}}
    }
  }
  child {node {adrain\vphantom{jd}}
    child {node {ad\vphantom{jd}}
      child {node {a\vphantom{jd}}}
      child {node {d\vphantom{jd}}}
    }
    child {node (s1) {ra\vphantom{jd}} }
    child {node (s2) {in\vphantom{jd}} }
  };

  \path (s1.south) edge[-latex,draw, bend left,dashed,my-grey] (t1.south);
  \path (s2.south) edge[-latex,draw, bend left,dashed,my-grey] (t2.south);

\end{tikzpicture}
  \caption{Example of a block tree for \(T=\texttt{abrainadrain}\) with \(s=2\) and \(\tau=3\). On the third level, only two children are created, as the block size is not dividable by three. We included this case for the sake of a small example. The pointer to marked blocks are given as dashed gray arrows.}
  \label{fig:example_block_tree}
\end{figure}
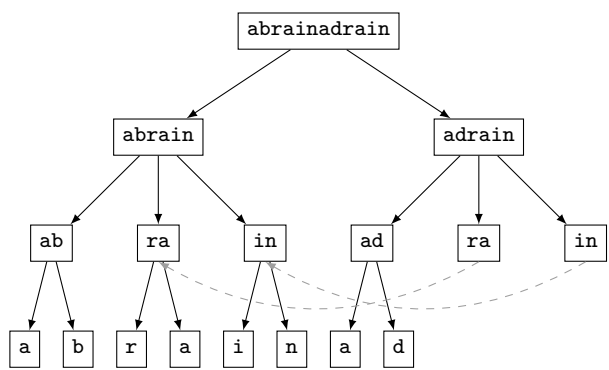

Interestingly, the number of blocks per level is bounded, which will later help us with analyzing the running time of our algorithms.
\begin{lemma}{\cite{BelazzouguiCGGK2021BlockTrees}}
  \label{lem:blocks_per_level}
  Any level of a block tree (except the first) contains at most \(3z\tau\) blocks.
\end{lemma}

We do not have to continue the marking of blocks until block have size one.
Instead, we stop as soon as explicitly storing the subtrees requires less space than storing the pointers and offsets for unmarked blocks \(B\), i.e., when \(|B|\in\Theta(\log_\sigma n)\).
Overall, the block tree requires \(O(s+z\tau\log_\tau\frac{n\log\sigma}{s\log n})\) words of space.
When choosing \(s=z\), the block tree requires \(O(z\tau\log\frac{n\log\sigma}{z\log n})\) words of space.
Note that there are different space-time trade-offs depending on the choice of \(s\) and \(\tau\)~\cite{BelazzouguiCGGK2021BlockTrees}.
When choosing \(s=\log_\sigma n\), we also get \(h=\log_\tau\frac{n\log\sigma}{s\log n}\).

\subsection{Sequential Block Tree Construction.}
\label{sec:sequential_block_tree_construction}
We first describe the original block tree construction algorithm~\cite{BelazzouguiCGGK2021BlockTrees}.
Note that there also exists a highly engineered block tree construction algorithm based on basic toolbox data structures for compression~\cite{KopplKM2023LPFBlockTrees}.
However, the algorithmic idea is the same and on a high level, it simulates the procedure described below.

The block tree is constructed in two phases followed by an optional pruning phase.
In the first phase, we mark all leftmost occurrences of consecutive blocks \(B_i\cdot B_{i+1}\).
Then, in the second phase, we compute the pointers (and offsets) of the unmarked blocks to their leftmost occurrence on the level.
During the pruning phase, we might reduce the number of nodes in the tree.
This phase does not improve the asymptotic size of the block tree.

\begin{description}
\item[First Phase:]
In the first phase, we use a hash table to identify the leftmost occurrence of each pair of consecutive blocks, i.e., when a pair is already contained in the hash table, it cannot be the leftmost occurrence.
Since we are not only interested in exact matches, we have to also scan the text to identify positions where such a pair may occur between blocks.
To efficiently compare blocks, we use Karp-Rabin fingerprints~\cite{KarpR1987Fingerprints}.
While we still have to verify that two matching fingerprints refer to the same substing, we only have to verify each substring once.
This results in a running time linear in the size of the level.
\item[Second Phase:]
In the second phase, we again use hash tables and Karp-Rabin fingerprints to identify blocks.
This time, we are interested in the unmarked blocks.
For those blocks, we now have to compute the pointer (and offset) to the leftmost occurrence, which will be contained in a pair of consecutive marked blocks.
To this end, we store the fingerprints of all unmarked blocks in the hash table.
Then, we scan the current level to find the leftmost occurrences of these blocks using their fingerprints.
We can do so in time linear in the size of the level, as we also have to consider each block only once.
\item[Optional Pruning:]
Finally, there is an optional pruning step, which can further reduces the size of the block tree in practice.
Unfortunately, there is no asymptotic improvement achievable in this step.
Remember that in the first phase, we mark all leftmost occurrences of consecutive pairs of blocks.
This is necessary to ensure that the leftmost occurrence of any unmarked block does exist on the same level (partially) in one (or two) marked blocks.
However, not all marked blocks are the target of a pointer.
During the pruning step, these untargeted but marked blocks are removed.
Note that we do not consider this step in our parallel algorithm described below, as it does not improve the asymptotic space bounds.
However, it is part of our implementation, cf.~\cref{sec:implementation_details}.
\end{description}

\subsubsection{Running Time}
The running time depends on the sizes of the level of the block tree.
Let \(b_k\) be the block size on the \(k\)-th level.
Furthermore, let \(c_k\) be the number of blocks on the \(k\)-th level.
This gives us \(b_k\cdot c_k\) characters to process on level \(k\).
According to \cref{lem:blocks_per_level}, there are at most \(3z\tau\) blocks per level.

Thus, levels starting after level  \(\ell\geq 1+\log_\tau\frac{3z}{s}\), we can bound the total number of character to process (on all levels from \(\ell\) to the last level \(h\)) by \(\sum_{m=\ell}^h b_m\cdot c_m\leq \sum_{m=\ell}^h\frac{n}{\tau^{m-\ell}}=n\sum_{m=\ell}^h\frac{1}{\tau^{m-\ell}}\leq \frac{n}{1-\tau^{=\ell}}=O(n)\).
All previous levels, contain at most \(O(n)\) characters (each).
Since we can process each character in constant expected time with high probability, this gives us the following running time.

\begin{lemma}[\cite{BelazzouguiCGGK2021BlockTrees}]
  The block tree of a text \(T\in[1,\sigma]^n\) can be constructed in \(O(n(1+\log_\tau\frac{z}{s}))\) expected time with high probability and \(O(s+z\tau\log_\tau\frac{n\log\sigma}{s\log n})\) working space.
\end{lemma}

\section{Parallel Block Tree Construction}
\label{sec:parallel_block_tree_construction}
Now, we present the main result of this paper: The first dedicated parallel block tree construction algorithms.
In \cref{sec:parallel_sort}, we discuss a simple parallelization based on sorting that requires \(O(n)\) words of working memory.
Then, in \cref{sec:memory_efficient_bt}, we give a more space-efficient construction algorithm that introduces a trade-off between scalability and working memory requirements.

\subsection{Parallel Sorting using \(O(n)\) Words of Memory.}
\label{sec:parallel_sort}
Our first parallel block tree construction algorithm uses sorting instead of hash tables.
We simply compute all fingerprints whenever necessary in parallel.

\begin{lemma}[\cite{Ellert0S2020LCPAwareStringSorting}]
  \label{lem:par_fingerprints}
 For any length-\(\ell\) substring \(T[i..i+\ell-1]\), the Karp-Rabin fingerprint \(\phi(i,i+\ell-1)\) can be computed in \(O(\log\ell)\) depth, \(O(\ell)\) work, and \(O(\ell)\) words of space in the EREW model.
\end{lemma}

Since we do need the fingerprints of all pairs of blocks (in the first phase) and all substrings of length \(b_\ell\) (in the second phase on each level \(\ell\)), we adopt \cref{lem:par_fingerprints}.

\begin{lemma}
  Computing the Karp-Rabin fingerprint of every length-\(\ell\) substring requires \(O(\log n)\) depth, \(O(n)\) work, and \(O(n)\) words of space in the EREW model.
\end{lemma}

The space necessary to store the fingerprints is also the space bottleneck of the construction algorithm, as we need it on all levels.

We now can efficiently compute the fingerprints, but we still have to validate all comparisons.
This can be done in \(O(\log n)\) time when reading memory cells is exclusive (EREW and CREW) and in \(O(1) \) time in the CRCW models.
In both cases, we require \(O(n)\) work and no additional space.

Then, we sort tuples consisting of the fingerprint and the corresponding block id.
In the result, matching fingerprints are consecutive.
Using \cref{lem:sorting_ew,lem:sorting_cw}, we can sort in \(O(\log n)\) time and \(O(\frac{\log n}{\log\log\log n})\) time in EREW and CRCW respectively.

Using the sorted fingerprints, we have to identify the leftmost occurrences.
Since all fingerprints are sorted (first by fingerprint and the by position), this can be done in \(O(1)\) time.

The second phase is a little bit more sophisticated.
Here, we first again use sorting to identify the block and position where we have to point to for the unmarked blocks.
Now, we combine the fingerprints obtained during the scan and the fingerprints of the unmarked blocks.
We sort these fingerprints, again as tuple of fingerprint and text position.
For our unmarked blocks, we use the text position to also indicate that they are an unmarked block.
After sorting, the fingerprints are grouped and sorted by text position with the unmarked blocks at the end.

We then build the array \(\mathit{occ}\) that will help us to compute the pointers and offsets.
For each fingerprint, we write a zero in \(\mathit{occ}\), if it is the same as the fingerprint to its left and its text position otherwise.
For the rightmost occurrence of a fingerprint, we write the negated text position of the first occurrence in \(\mathit{occ}\).
Then, after computing the prefix sum over \(\mathit{occ}\), it contains the text positions for all but the rightmost fingerprint (for which we find the correct position to its left).
The prefix sum can be computed in \(O(\log n)\) time and \(O(\frac{\log n}{\log\log n})\) time (Common-CRCW), respectively.

This allows us to compute every level of the block tree using only sorting and prefix sums.
In the description above, we always assume that we have to operate on \(O(n)\) objects all the time.
However, on each level of the block tree, we might replace some blocks with references.

\begin{theorem}
  The block tree of a text \(T\in[1,\sigma]^n\) can be computed in \(O(\log n\log_\tau\frac{n\log\sigma}{s\log n})\) time and \(O(n\log n(1+\log_\tau\frac{n\log\sigma}{s\log n}))\) work in \(O(n)\) words of space in EREW PRAM.
\end{theorem}
\begin{proof}
  We know that there are at most \(3z\tau\) blocks on all levels of the block tree but the first (\cref{lem:blocks_per_level}).
  Therefore, there is a level \(\ell>1+\log_\tau\frac{3\tau}{z}\), such that the total block length at all following level is linear in the text length.
  For all previous level, i.e., level \(<\ell\) we have to assume that each level has size \(O(n)\).

  We have to process each level individually, as we require information about marked and unmarked block on preceding levels.
  The total length of the remaining levels is \(\sum_{m=\ell}^hk_m\cdot c_m\leq \sum_{m=\ell}^h\frac{n}{\tau^{m-\ell}}\).

  On each level, we can do all operations in logarithmic time (and linear work) on the size of that level.
  This leaves us with total time \(\sum_{m=\ell}^h\log n-\log\tau^{m-\ell}\leq\sum_{m=\ell}^h\log n\).

  We know that the height of a block tree is \(h=\log_\tau\frac{n\log\sigma}{s\log n}\).
  The total time necessary for the last level is therefore \(O(\log_\tau\frac{n\log\sigma}{z\log n}\log n)\).
  If we include the first level, too, we get a total time of \(O(\log n(\log_\tau\frac{z}{s}+\log_\tau\frac{n\log\sigma}{z\log n}))=O(\log n\log_\tau\frac{n\log\sigma}{s\log n})\).
\end{proof}

When considering Common-CRCW, we can improve the result using better sorting and prefix sum algorithms.
\begin{lemma}
  In the Common-CRCW, the block tree of a text \(T\in[1,\sigma]^n\) can be computed in \(O(\frac{\log n}{\log\log n}\log_\tau\frac{n\log\sigma}{s\log n})\) time and \(O(n\log n(1+\log_\tau\frac{n\log\sigma}{s\log n}))\) work in \(O(n)\) words of space.
\end{lemma}

\subsection{Domain Decomposition using \(O(s+K(z\tau))\) Words of Memory.}
\label{sec:memory_efficient_bt}
The big disadvantage of the algorithm described above is that it requires a lot of working memory.
We tackle this problem, by utilizing \emph{domain decomposition} to introduce a trade-off between scalability and required working memory.
The idea of domain decomposition is to partition the input and work on each partition as independent as feasible.

For block trees, each level heavily depends on the previous ones.
Therefore, we have to merge and repetition the input on each level.
The algorithm presented in \cref{lem:par_fingerprints} remains nearly the same.
The main difference being that we now pre-filter locally.

Now, our algorithm uses \(K\) partitions.
For each partition, we have a hash table that serves as local filter.
Now, only the locally leftmost occurrences of block pairs (and in the second phase substrings obtained during the scan) will be sorted.
This significantly reduces the memory overhead, as we now only consider a \(K\)-approximation of the blocks are necessary for the construction.

\begin{theorem}
  For any integer \(K>0\), the block tree of a text \(T\in[1,\sigma]^n\) can be computed in \(O(\frac{n}{K}\log n\log_\tau\frac{n\log\sigma}{s\log n})\) time and \(O(n\log n(1+\log_\tau\frac{n\log\sigma}{s\log n}))\) work in \(O(s+K(z\tau))\) words of space in EREW PRAM.
\end{theorem}
\begin{proof}
  The argument for the running time is similar as in the previous approach.
  The main difference being that on each level, we now have to filter locally introducing a new bottleneck.

  However, filtering helps us to significantly reduce the amount of working memory necessary during construction.
  Since we now can control how often a fingerprint occurs overall while finding the leftmost occurrence.
  Therefore, the required space is the same as during the sequential construction, i.e., \(O(s+z\tau)\)~\cite{BelazzouguiCGGK2021BlockTrees} for \emph{each} partition.
  Because in the worst case, the fingerprint occurs in all partitions at least once.
\end{proof}

\begin{lemma}
  For any integer \(K>0\), in the Common-CRCW, the block tree of a text \(T\in[1,\sigma]^n\) can be computed in \(O(\frac{n}{K}\frac{\log n}{\log\log n}\log_\tau\frac{n\log\sigma}{s\log n})\) time and \(O(n\log n(1+\log_\tau\frac{n\log\sigma}{s\log n}))\) work in \(O(s+K(z\tau))\) words of space.
\end{lemma}

\subsection{Data-Parallel Karp-Rabin Fingerprints}
\label{sec:simd_karp_rabin}

We accelerate the computation of Karp-Rabin fingerprints using SIMD techniques.
An implementation was initially designed for the Intel AVX2 instruction set extension and later ported to AArch64~ASIMD.
We use \(q=2^{32}\), \( r=33\), \(\sigma=256\) (like the djb2 hash function).

As AVX2~vectors are 32~bytes long, they can hold 8~hashes of 32~bits each.
In each iteration of our function, 16~new hashes~\(\phi(s..e)\), \(\phi(s+1,e+1)\), \dots, \(\phi(s+15,e+15)\) stored in 2~vectors are computed based the value of the previous 16~hashes according to the scheme
\begin{equation}
\phi(s\,..\,e)=\phi(s-16\,..\,e-16)\cdot r^{16}-\phi(s-16\,..\,s-1)\cdot r^{e-s+1}+\phi(e-15\,..\,e).
\end{equation}

\begin{figure}
\centering
\includegraphics[width=\textwidth]{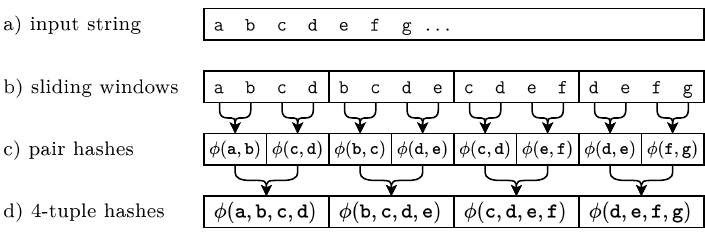}
\caption{Computing the hashes of 4-tuples using 1~permutation and 2~pairwise multiply-adds}
\label{fig:karp_rabin_hash4}
\end{figure}

The 16-tuple hashes \(\phi(e-15\,..\,e)\) to \(\phi(e\,..\,e+15)\) are computed in two steps.
First, hashes of 4-tuples are computed as shown in \cref{fig:karp_rabin_hash4}:
A chunk of input~(a) is loaded into a SIMD~register and permuted to form sliding windows~(b) of 4~characters each.
Using a data-parallel pairwise multiply-add instruction, pairs of characters~\((T[i],T[i+1])\) are multiplied with~\(( r,1)\) to give pair hashes~(c) following~\(\phi(i,i+1)=T[i]\cdot r+T[i+1]\).
This step is repeated by multiply-adding pairs of pair hashes with \(( r^2,1)\), producing 4-tuple hashes~(d).
As~$ r<256$, the pair hashes fit into 16~bits, preserving vector length through steps~(a) to~(d).

\begin{figure}
\centering
\includegraphics[width=\textwidth]{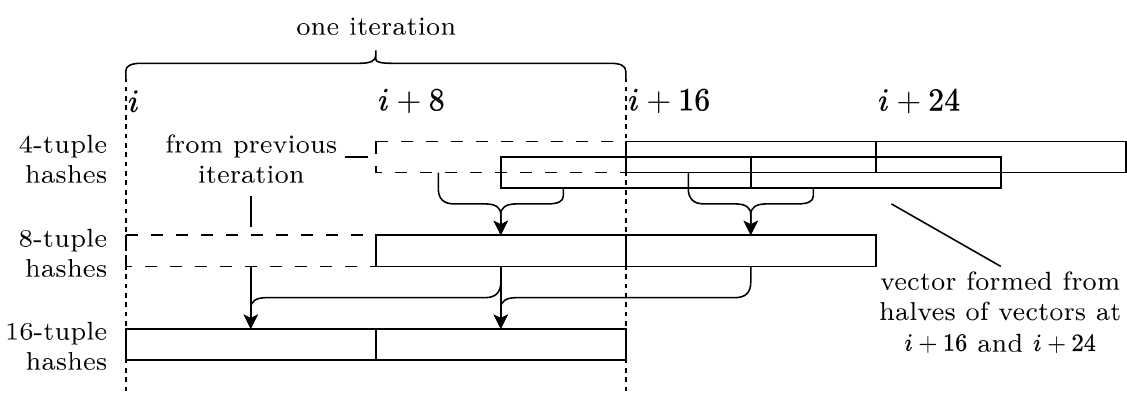}
\caption{Computing 16~hashes of 16-tuples in vectors of 8~hashes each}
\label{fig:karp_rabin_pruned}
\end{figure}

These 4-tuple hashes are then combined into 8-tuple hashes and finally 16-tuple hashes according to \cref{fig:karp_rabin_pruned}, using elementwise operations.
Some intermediate values are carried over from the previous iteration to reduce the number of operations.
Vectors of 4-tuple hashes offset by~4 needed to compute the 8-tuple hashes are formed by joining the rear half of the previous vector with the front half of the next vector.

The 16-character hashes are cached in a ring buffer, so they can be subtracted once the start of the window reaches the current end.
For~17--32 characters, a variant of the above algorithm is used, keeping the ring buffer in vector registers to avoid costly store-reload latency.
For hashes of up to 16~characters, dedicated routines are used to directly compute each hash from the value of the characters involved.

The total computational load per iteration comprises 4~pairwise multiply-add operations, 8~multiplications, 8~additions, and 6 permutations per iteration.
However, due to high data-parallelism and the high latency of 32-bit SIMD multiplications on current Intel microarchitectures, throughput is neverthless latency-bound at 11~cycles per 16~characters (10~cycles to compute the product \(\phi(s-16\,..\,e-16)\cdot r^{16}\), 1~cycle to add the remaining terms)~\cite{Abel19a}.
This compares favorably with the \(\mathbin{\sim}2.5\)~cycles per character throughput of computing Karp-Rabin fingerprints conventionally,
i.\,e.~through the relation~\cite{KarpR1987Fingerprints}
\begin{equation*}
\phi(s\,..\,e)=\phi(s-1\,..\,e-1)\cdot r-T[s-1]\cdot r^{e-s+1}+T[e],
\end{equation*}
and yields a 5\,\%~improvement in performance over all tested configurations of our algorithm.

\subsection{Further Implementation Details.}
\label{sec:implementation_details}

We implemented our block tree in C++ using the same interface as previous block tree implementations~\cite{BelazzouguiCGGK2021BlockTrees,KopplKM2023LPFBlockTrees}.
A practical implementation of a block tree consists of only one bit vector per level (to mark inner nodes) as well as pointers and offsets to earlier occurrences for inner nodes.
We implemented the memory efficient algorithm based on domain decomposition that we present in \cref{sec:memory_efficient_bt}, as the number of PEs we can use is determined by the hardware and usually significantly smaller than the input size.
Preliminary experiments showed that, any number of partitions not equal to the number of PEs is detrimental to the running time of the algorithm.

In our implementation, we are not sorting.
While sorting eases the analysis of the algorithm, it is not necessary in practice.
Instead, each PE is responsible for a local partition of the current level.
Initially, the blocks are only marked regarding occurrences in the local partition, i.e., we use a hash table (\texttt{ankerl::unordered\_dense::map}) as filter to only consider the leftmost occurrence locally.

These marked blocks are then communicated to the responsible PE\@.
To this end, we use a wait-free multiple producer single consumer (MPSC) queue~\cite{AdasF2020Jiffy}.
Since the fingerprints are random, they are used to determine the PE that collects the blocks by a modulo operation.
Note that in practice, this is very balanced and there is no imbalance.
As soon as we have marked all blocks globally, we can continue with the second phase in the same fashion.

Another interesting observation is that the size of the MPSC queue does not affect the running time of the algorithm, after a certain size.
In our experimental evaluation, in \cref{sec:experimental_evaluation}, we use a queue of size 512 fingerprints.
Smaller queues led to imbalances between different PEs, but, since there were no imbalances with this size, larger queues did not result in any measurable improvement.

\section{Experimental Evaluation}
\label{sec:experimental_evaluation}
We conducted our experiment on a server equipped with an AMD EPYC 7713 CPU (64 physical cores with hyperthreading support running at 2.0\,GHz base frequency with 3.66\,GHz turbo boost, 64\,KB L1 and 512\,KB L2 cache per core and 256\,MB shared L3 cache) and 1024\,GB DDR4 RAM.
The server is running Ubuntu 20.04 (kernel version 5.15.0).
All code has been compiled using the GNU Compiler Collection (GCC) version 11.4.0 using the provided build scripts.

\newcommand{\lpfbt}{\texttt{LPF-BT}}
\newcommand{\parbt}{\texttt{Par-BT}}
\newcommand{\parlz}{\texttt{Par-LZ}}
\newcommand{\pigz}{\texttt{pigz}}
\newcommand{\pxz}{\texttt{pxz}}

We compare our new algorithm (\parbt, see \cref{sec:parallel_block_tree_construction}) with the state-of-the-art block tree construction algorithm (\lpfbt)~\cite{KopplKM2023LPFBlockTrees}.\footnote{\url{https://github.com/pasta-toolbox/block_tree}; archived at Zenodo~\cite{PASTA_BLOCK_TREE_ARTIFACT}, last accessed 2024-10-09.}
Since there are no other direct competitors, we also include a state-of-the-art parallel LZ77 compression algorithm (\parlz)~\cite{ShunZ2013ParallelLZ},\footnote{\url{https://github.com/zfy0701/Parallel-LZ77}, last accessed 2024-10-09.}
the parallel LZMA compressor \pxz~\cite{Jnovy2024pxz}, and a parallel gzip implementation \pigz~\cite{Adler2024pigz}.
Note that the the latter three algorithms only compress the input---without any support for random access and rank/select queries.
However, they give us insights in a similar problem, as block trees are a LZ77 approximation with additional query support, i.e., on the one hand they relax the problem but on the other hand they require additional information.
We do not include the implementation of the original block tree construction algorithm~\cite{BelazzouguiCGGK2021BlockTrees}, as it is only sequential and up to an order of magnitude slower than \lpfbt~\cite{KopplKM2023LPFBlockTrees}.
We used the fastest configuration of \lpfbt\ and chose the same parameters (arity) for \parbt, ensuring that both algorithms compute the same block tree.

As inputs, we are using the real-world repetitive texts of the Pizza\&Chili corpus\footnote{\url{https://pizzachili.dcc.uchile.cl/repcorpus.html}, last accessed 2024-10-09.}, which is commonly used in the field of text indexing.
The same and very similar (but smaller) inputs have been used by the already existing algorithms used in this evaluation~\cite{KopplKM2023LPFBlockTrees,ShunZ2013ParallelLZ}.
Details on the inputs are given in \cref{tab:text_details}.

\begin{table}
  \caption{Name, size \(n\), alphabet size \(\sigma\), number of LZ77 factors \(z\), and average LZ77 factor length \(\ceil{n/z}\) of the inputs used in the experimental evaluation.}
  \label{tab:text_details}
  \centering
  \begin{tabular}{lrrrr}
    \toprule
    Name & \(n\) & \(\sigma\) & \(z\) & \(\ceil{n/z}\)\\
    \midrule
    cere & 461286644 & 5 & 1700630 & 272 \\
    coreutils & 205281778 & 236 & 1446468 & 142 \\
    einstein.de & 92758441 & 117 & 34572 & 2684\\
    einstein.en & 467626544 & 139 & 89467 & 5227\\
    escherichia & 112689515 & 15 & 2078512 & 55\\
    influenza & 154808555 & 15  & 769286 & 202\\
    kernel & 257961616 & 160 & 793915 & 325\\
    para & 429265758 & 5 & 2332657 & 185\\
    world\_leaders & 46968181 & 89 & 175740 & 268\\
    \bottomrule
  \end{tabular}
\end{table}

\begin{figure*}[h!]
  \centering
  \input{figures/throughput.tex}

  \caption{Throughput (MiB/s) of the parallel block tree construction algorithms and the parallel general purpose compression algorithms (dashed lines) in a strong scaling experiment.}
  \label{fig:throughput}
\end{figure*}
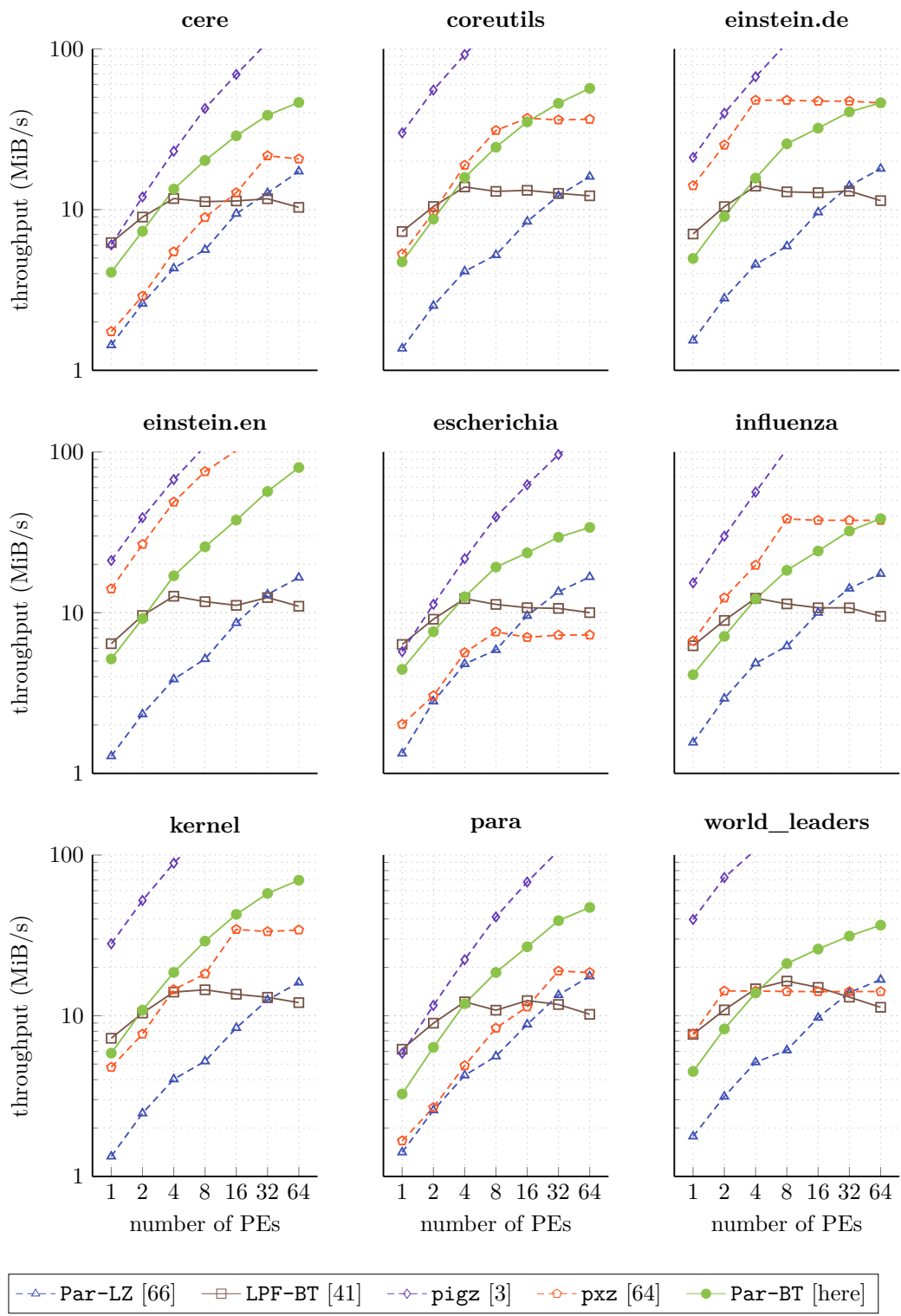

\begin{figure*}[h!]
  \centering
  \input{figures/memory.tex}
  \caption{Space-overhead (\%) of the parallel block tree construction algorithms and the parallel general purpose compression algorithms (dashed lines) in a strong scaling experiment. Missing data points refer to no heap allocation.}
  \label{fig:memory}
\end{figure*}
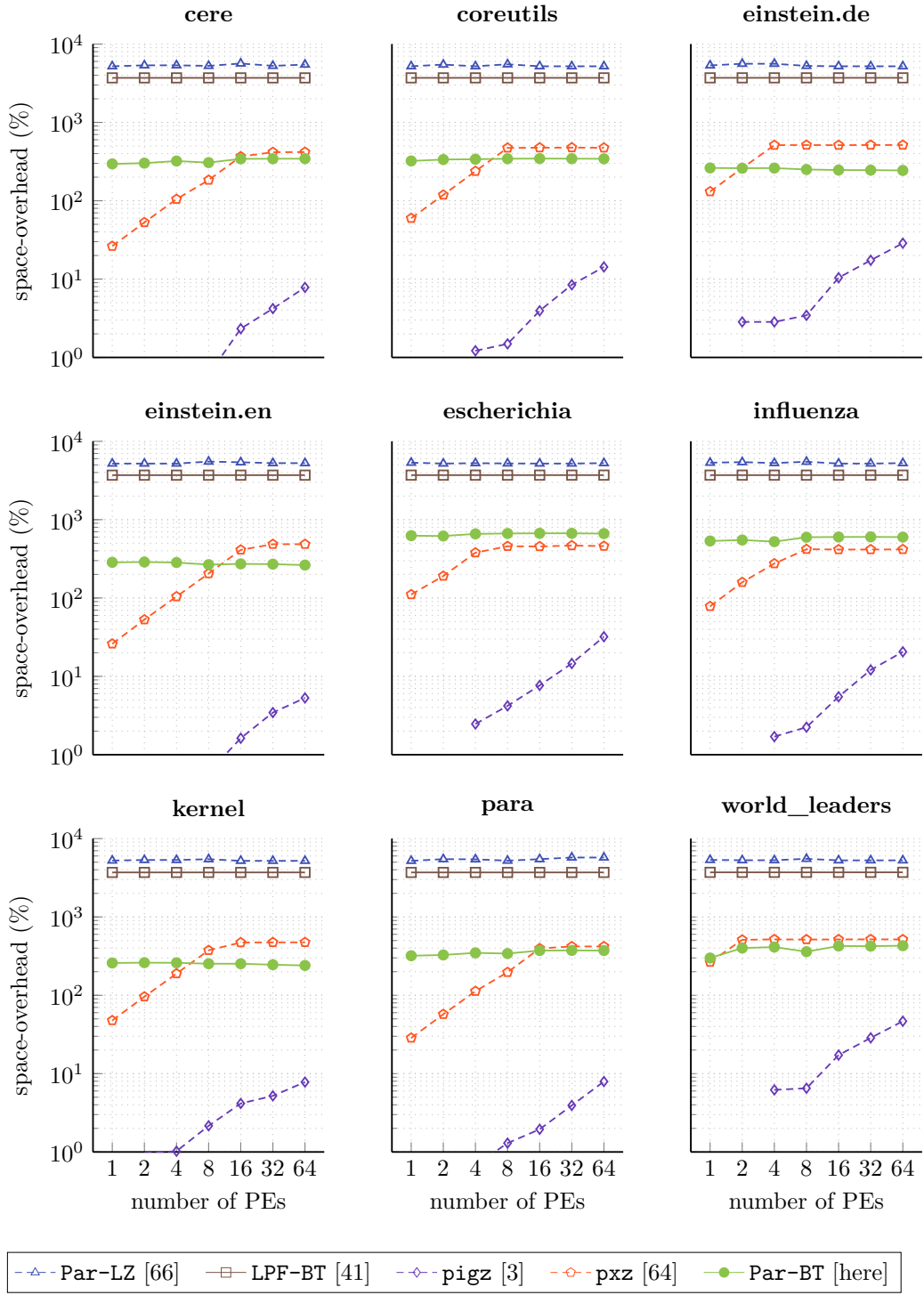

\subsection{Scalability}
\label{sec:evaluation_scalability}
We first consider the scalability of all implementations.
To this end, we conducted a strong scaling experiment where we use one physical CPU core per PE\@.
The throughput, i.e., processed input in MiB per second is depicted in \cref{fig:throughput}.

Overall, all tested algorithms behave similarly on all inputs, independently of the input size, alphabet size, and average Lempel-Ziv factor length.
We included \parlz\ to show the difficulty of efficiently computing the LZ77 factorization in parallel.
Therefore, it comes at no surprise that \parlz\ is the slowest of all tested algorithms.
Interestingly, it scales very similar to \parbt, despite it computing an LZ77 approximation and additional information required to answer queries instead of the LZ77 factorization.

Most notably \pigz\ is the fastest and best scaling algorithm overall.
However, this comes at the cost of compression ratio compared to \pxz, the third general purpose compression algorithms we include in our benchmark.
Overall, \pxz achieves speed similar to \parbt.
On some instances it is faster, on some ov very similar speed, and on some slower.
Surprisingly, it does not scale as well as \pigz.
We want to highlight again that these algorithms do not and cannot easily be used for block tree construction and data compressed using them does not provide any query support.

Currently, \lpfbt\ is the state-of-the-art block tree construction algorithm when it comes to construction speed.
The authors of \lpfbt\ also present a naive parallelization of their algorithm, which is---to the best of our knowledge---the only other parallel block tree construction algorithm.
Their parallelization is based on known parallel algorithms for basic text compression data structures.
Since only the construction of these data structures and not the whole construction is parallelized, \lpfbt\ does not scale very well, i.e., there is no more speedup when using more than 8 PEs.

Our new algorithm \parbt\ on the other hand achieves speedup of up to 15.5 (on einstein.de) andn of 11.1 on average using 64 PEs.
While this by no means a perfect speedup it is what we can expect regarding the underlying problem (LZ77 factorization) and the scalability of algorithms for that problem (\parlz).

Note that even when using one PE, the memory throughput that we are able to achieve is significantly greater than the throughput we achieve when constructing block trees.
To be more precise, the throughput is always greater than 1 GiB/s when using one thread, which is multiple order of magnitude greater than the construction throughput.
Hence, the additional memory bandwidth when using multiple threads should be of no benefit and the algorithms are CPU-bound.

\subparagraph*{COST.}
The configuration that outperforms a single thread (COST)~\cite{McSherryIM2015Cost} is the number of PEs necessary for the parallel algorithm to be faster than the fastest sequential algorithm.
When not analyzing the COST of a parallel algorithms one could simply use a very slow algorithm with great scalability and provide impressive speedups.
In our case, the COST of \parbt\ is 2, i.e., two PEs are sufficient to be faster than the fastest sequential block tree construction algorithm (\lpfbt\ with one PE).
This is the best we can achieve since our algorithms is not the fastest sequential algorithm.

\subsection{Memory Requirements}
\label{sec:memory_requirements}
In addition to the construction speed, we are also interested in the amount of memory required during construction, as requiring too much working memory seriously hinders the applicability of the algorithms in scenarios where memory is scarce, e.g., in HPC clusters or on shared hardware in general.
We look at therein memory allocated on the heap, in percent of the input size in \cref{fig:memory}.
To the best of our knowledge, no implementation allocates significant amounts of memory on the stack, hence only considering the heap is sufficient.

While being very fast, \lpfbt\ requires a lot of memory, as it utilizes a lot of toolbox data structures that need \(O(n\log n)\) bits for a length-\(n\) input.
The same is true for \parlz.
Interestingly, \lpfbt\ requires the same amount of working memory independently of the number of PEs used.
As expected, \parbt\ requires less space than \lpfbt; its space-overhead is around an order of magnitude less than \parbt's.
On all inputs, \parlz\ requires the most memory during construction.
Both \pigz\ and \pxz\ are very memory efficient but require working memory proportional to the number of PEs.
Most notably, for low PE counts, \pigz\ does not require any heap allocations.
However, as mentioned before, this also results in worse compression ratios compared to \pxz.

There are slight differences in space-overhead by \parbt\ depending on the number of PEs.
This can be explained by the local hash tables, which are filled depending on the fingerprints of the blocks.
Since our hash function is only pseudo-random, we cannot guarantee a perfect partitioning of all incoming fingerprints.
However, as we can see, it is sufficient to reasonably balance hash table sizes.

\section{Conclusion and Future Work}
\label{sec:conclusion_and_future_work}
In this paper, we present the first dedicated parallel block tree construction algorithms.
Our algorithms provide a good trade-off between required memory during construction and scalability.
The C++ implementation is the best scaling block tree construction algorithm.
Additionally, as soon as we use two threads, our implementation is faster than the state-of-the-art sequential algorithm.
We achieve all this while requiring around 10 times less working memory than our competitors during construction.

Still, the block tree is not yet a drop-in replacement for wavelet trees in all scenarios.
While block trees compress well on highly repetitive inputs, they are at most as small as wavelet trees on non-repetitive inputs.
We will tackle this problem by using better compression and more space-efficient rank and select data structures for bit vectors, e.g.,~\cite{KurpiczRS2025TheoryMeetsPracticeRankAndSelect}.
Furthermore, there is still a huge space-overhead for auxiliary data structures, when rank and select queries have to be answered.
In the future, we plan to tackle these problems and making block trees strictly superior to wavelet trees.

\bibliography{literature}

\end{document}

%% file: figures/throughput.tex
\begin{tikzpicture}
  \begin{groupplot}[
    group style={
      group size= 3 by 3,
      horizontal sep=1cm,
      vertical sep=1.25cm,
    },
    width=5cm,
    height=6.5cm,
    xmode=log,
    ymode=log,
    log basis x={2},
    log basis y={10},
    xtick={1,2,4,8,16,32,64},
    xticklabels={1,2,4,8,16,32,64},
    yticklabels={1,10,100},
    ymin=1,
    ymax=100,
    cycle list name=colorList,
    every axis title/.append style={font=\bfseries}
    ]
    \nextgroupplot[
    title={cere},
    ylabel={throughput (MiB/s)},
    xmajorticks=false,
    ]
    \addplot coordinates { (1,1.43464) (2,2.60148) (4,4.3145) (8,5.6264) (16,9.43579) (32,12.6879) (64,17.3005) };
    \addlegendentry{algo=plz77}
    
    \addplot coordinates { (1,6.23367) (2,8.99176) (4,11.7228) (8,11.2112) (16,11.3234) (32,11.6853) (64,10.3281) };
    \addlegendentry{algorithm=LPF-pruning-1}

    \addplot coordinates { (1,6.03428) (2,11.9962) (4,23.0984) (8,42.609) (16,69.4211) (32,109.901) (64,161.395) };
    \addlegendentry{algo=pigz}
    
    \addplot coordinates { (1,1.74071) (2,2.89648) (4,5.46765) (8,8.94132) (16,12.7497) (32,21.6462) (64,20.6743) };
    \addlegendentry{algo=pxz}

    \addplot coordinates { (1,4.07482) (2,7.3373) (4,13.4183) (8,20.2372) (16,28.794) (32,38.6835) (64,46.4882) };
    \addlegendentry{algo=rec\_shard}
    
    \legend{}

    \nextgroupplot[
    title={coreutils},
    ticks=none,
    ]
    \addplot coordinates { (1,1.36842) (2,2.52346) (4,4.13136) (8,5.22088) (16,8.43243) (32,12.2066) (64,16.0494) };
    \addlegendentry{algo=plz77}
    
    \addplot coordinates { (1,7.30659) (2,10.4562) (4,13.8406) (8,12.9885) (16,13.1857) (32,12.6314) (64,12.1885) };
    \addlegendentry{algorithm=LPF-pruning-1}

    \addplot coordinates { (1,30.0288) (2,55.484) (4,92.2778) (8,145.836) (16,192.59) (32,213.827) (64,202.804) };
    \addlegendentry{algo=pigz}
    
    \addplot coordinates { (1,5.2942) (2,9.59064) (4,18.8954) (8,31.1179) (16,37.0361) (32,36.183) (64,36.5139) };
    \addlegendentry{algo=pxz}

    \addplot coordinates { (1,4.73126) (2,8.73284) (4,15.8463) (8,24.4652) (16,35.1415) (32,45.8689) (64,56.8803) };
    \addlegendentry{algo=rec\_shard}
    
    \legend{}

    \nextgroupplot[
    title={einstein.de},
    ticks=none,
    ]
    \addplot coordinates { (1,1.53444) (2,2.80255) (4,4.55959) (8,5.92593) (16,9.62801) (32,14.0631) (64,17.9867) };
    \addlegendentry{algo=plz77}
    
    \addplot coordinates { (1,7.05241) (2,10.4401) (4,14.0021) (8,12.9032) (16,12.7569) (32,13.0491) (64,11.375) };
    \addlegendentry{algorithm=LPF-pruning-1}

    \addplot coordinates { (1,21.1375) (2,39.8001) (4,67.1523) (8,108.572) (16,144.423) (32,172.9) (64,164.748) };
    \addlegendentry{algo=pigz}
    
    \addplot coordinates { (1,14.102) (2,25.2375) (4,47.974) (8,47.9699) (16,47.3271) (32,47.3278) (64,46.0973) };
    \addlegendentry{algo=pxz}

    \addplot coordinates { (1,4.96663) (2,9.06166) (4,15.7017) (8,25.6691) (16,32.1227) (32,40.5998) (64,46.2185) };
    \addlegendentry{algo=rec\_shard}
    
    \legend{}

    \nextgroupplot[
    title={einstein.en},
    ylabel={throughput (MiB/s)},
    xmajorticks=false,
    ]
    \addplot coordinates { (1,1.27965) (2,2.33596) (4,3.85281) (8,5.15195) (16,8.62821) (32,12.9927) (64,16.5428) };
    \addlegendentry{algo=plz77}
    
    \addplot coordinates { (1,6.41257) (2,9.59719) (4,12.6591) (8,11.712) (16,11.1088) (32,12.4178) (64,10.9682) };
    \addlegendentry{algorithm=LPF-pruning-1}

    \addplot coordinates { (1,21.1056) (2,39.0193) (4,67.427) (8,108.606) (16,144.942) (32,178.368) (64,174.861) };
    \addlegendentry{algo=pigz}
    
    \addplot coordinates { (1,14.0717) (2,26.6394) (4,48.8015) (8,75.5476) (16,101.928) (32,170.392) (64,170.403) };
    \addlegendentry{algo=pxz}

    \addplot coordinates { (1,5.15242) (2,9.18008) (4,16.9624) (8,25.7114) (16,37.7047) (32,56.8055) (64,80.0468) };
    \addlegendentry{algo=rec\_shard}
    
    \legend{}

    \nextgroupplot[
    title={escherichia},
    ticks=none,
    ]
    \addplot coordinates { (1,1.32952) (2,2.80472) (4,4.79283) (8,5.86301) (16,9.57494) (32,13.4931) (64,16.6927) };
    \addlegendentry{algo=plz77}
    
    \addplot coordinates { (1,6.34206) (2,9.10464) (4,12.2244) (8,11.2572) (16,10.7376) (32,10.6312) (64,10.0094) };
    \addlegendentry{algorithm=LPF-pruning-1}

    \addplot coordinates { (1,5.72491) (2,11.2689) (4,21.6632) (8,39.4496) (16,62.5206) (32,96.3343) (64,136.183) };
    \addlegendentry{algo=pigz}
    
    \addplot coordinates { (1,2.02034) (2,3.03942) (4,5.63679) (8,7.5995) (16,7.0257) (32,7.25212) (64,7.26564) };
    \addlegendentry{algo=pxz}

    \addplot coordinates { (1,4.43749) (2,7.61796) (4,12.5546) (8,19.223) (16,23.5644) (32,29.4949) (64,33.9332) };
    \addlegendentry{algo=rec\_shard}
    
    \legend{}

    \nextgroupplot[
    title={influenza},
    ticks=none,
    ]
    \addplot coordinates { (1,1.55679) (2,2.92392) (4,4.83553) (8,6.18947) (16,10.0) (32,14.1414) (64,17.4533) };
    \addlegendentry{algo=plz77}
    
    \addplot coordinates { (1,6.22716) (2,8.94746) (4,12.2987) (8,11.3393) (16,10.7227) (32,10.7137) (64,9.49413) };
    \addlegendentry{algorithm=LPF-pruning-1}

    \addplot coordinates { (1,15.3185) (2,29.9137) (4,56.2897) (8,104.237) (16,161.604) (32,231.901) (64,251.908) };
    \addlegendentry{algo=pigz}
    
    \addplot coordinates { (1,6.65577) (2,12.3554) (4,19.7536) (8,38.3162) (16,37.5763) (32,37.5767) (64,37.5811) };
    \addlegendentry{algo=pxz}

    \addplot coordinates { (1,4.11218) (2,7.12123) (4,12.1588) (8,18.3475) (16,24.1886) (32,32.2421) (64,38.4716) };
    \addlegendentry{algo=rec\_shard}
    
    \legend{}

    \nextgroupplot[
    title={kernel},
    xlabel={number of PEs},
    ylabel={throughput (MiB/s)}
    ]
    \addplot coordinates { (1,1.33153) (2,2.4705) (4,4.02619) (8,5.20635) (16,8.3959) (32,12.535) (64,16.1047) };
    \addlegendentry{algo=plz77}
    
    \addplot coordinates { (1,7.24126) (2,10.3875) (4,14.0311) (8,14.4998) (16,13.6081) (32,13.026) (64,12.0641) };
    \addlegendentry{algorithm=LPF-pruning-1}

    \addplot coordinates { (1,28.0441) (2,52.1564) (4,89.0114) (8,139.472) (16,187.415) (32,220.924) (64,194.766) };
    \addlegendentry{algo=pigz}
    
    \addplot coordinates { (1,4.77871) (2,7.69445) (4,14.5307) (8,18.2156) (16,34.4375) (32,33.3787) (64,34.1969) };
    \addlegendentry{algo=pxz}

    \addplot coordinates { (1,5.85063) (2,10.8265) (4,18.6057) (8,29.1297) (16,42.7287) (32,57.6822) (64,69.723) };
    \addlegendentry{algo=rec\_shard}
    
    \legend{}

    \nextgroupplot[
    title={para},
    xlabel={number of PEs},
    ymajorticks=false,
    ]
    \addplot coordinates { (1,1.41278) (2,2.58044) (4,4.25598) (8,5.59508) (16,8.82893) (32,13.4539) (64,17.6103) };
    \addlegendentry{algo=plz77}
    
    \addplot coordinates { (1,6.18708) (2,8.98452) (4,12.2032) (8,10.8233) (16,12.4184) (32,11.7465) (64,10.2051) };
    \addlegendentry{algorithm=LPF-pruning-1}

    \addplot coordinates { (1,5.8409) (2,11.5884) (4,22.3139) (8,41.1956) (16,67.9317) (32,106.375) (64,161.062) };
    \addlegendentry{algo=pigz}
    
    \addplot coordinates { (1,1.66523) (2,2.6831) (4,4.88781) (8,8.37312) (16,11.3594) (32,18.9908) (64,18.5553) };
    \addlegendentry{algo=pxz}

    \addplot coordinates { (1,3.2591) (2,6.35661) (4,11.8786) (8,18.6017) (16,26.8545) (32,39.0677) (64,47.2096) };
    \addlegendentry{algo=rec\_shard}
    
    \legend{}

    \nextgroupplot[
    title={world\_leaders},
    xlabel={number of PEs},
    legend to name={leg:throughput},
    legend columns=5,
    ymajorticks=false,
    ]
    \addplot coordinates { (1,1.77419) (2,3.14286) (4,5.13869) (8,6.09841) (16,9.72913) (32,13.9241) (64,16.7619) };
    \addlegendentry{\parlz~\cite{ShunZ2013ParallelLZ}}
    
    \addplot coordinates { (1,7.65883) (2,10.8709) (4,14.7108) (8,16.4302) (16,15.0107) (32,13.0477) (64,11.2849) };
    \addlegendentry{\lpfbt~\cite{KopplKM2023LPFBlockTrees}}

    \addplot coordinates { (1,39.6985) (2,72.3959) (4,107.927) (8,142.909) (16,211.938) (32,212.136) (64,171.014) };
    \addlegendentry{\pigz~\cite{Adler2024pigz}}
    
    \addplot coordinates { (1,7.69816) (2,14.256) (4,14.257) (8,14.1405) (16,14.1372) (32,14.141) (64,14.1409) };
    \addlegendentry{\pxz~\cite{Jnovy2024pxz}}

    \addplot coordinates { (1,4.49794) (2,8.27145) (4,13.8615) (8,21.1158) (16,25.9894) (32,31.3446) (64,36.6209) };
    \addlegendentry{\parbt\ [here]}
    
  \end{groupplot}
\end{tikzpicture}

\vspace{.5cm}
\begin{centering}
  \ref*{leg:throughput}
\end{centering}

%% file: figures/memory.tex
\begin{tikzpicture}
  \begin{groupplot}[
    group style={
      group size= 3 by 3,
      horizontal sep=1cm,
      vertical sep=1.25cm,
    },
    width=5cm,
    height=6.25cm,
    xmode=log,
    ymode=log,
    log basis x={2},
    log basis y={10},
    xtick={1,2,4,8,16,32,64},
    xticklabels={1,2,4,8,16,32,64},
    ymin=1,
    ymax=10000,
    cycle list name=colorList,
    every axis title/.append style={font=\bfseries}
    ]
    \nextgroupplot[
    title={cere},
    ylabel={space-overhead (\%)},
    xmajorticks=false,
    ]
    \addplot coordinates { (1,5201.02) (2,5350.67) (4,5351.34) (8,5263.38) (16,5645.39) (32,5267.47) (64,5467.64) };
    \addlegendentry{algo=plz77}

    \addplot coordinates { (1,3700.97) (2,3701.01) (4,3700.98) (8,3701.05) (16,3701.09) (32,3701.15) (64,3701.22) };
    \addlegendentry{algorithm=LPF-pruning-1}

    \addplot coordinates { (1,0.0) (2,0.0) (4,0.0) (8,0.600928) (16,2.32632) (32,4.20086) (64,7.83071) };
    \addlegendentry{algo=pigz}

    \addplot coordinates { (1,26.3242) (2,52.8428) (4,104.777) (8,183.697) (16,366.613) (32,416.487) (64,417.562) };
    \addlegendentry{algo=pxz}

    \addplot coordinates { (1,295.184) (2,302.093) (4,321.243) (8,306.814) (16,342.511) (32,343.576) (64,344.13) };
    \addlegendentry{algo=rec\_shard}
    \legend{}
    
    \nextgroupplot[
    title={coreutils},
    ticks=none,
    ]
    \addplot coordinates { (1,5202.41) (2,5468.88) (4,5205.4) (8,5513.52) (16,5202.42) (32,5202.62) (64,5202.97) };
    \addlegendentry{algo=plz77}
    
    \addplot coordinates { (1,3702.17) (2,3702.21) (4,3702.33) (8,3702.33) (16,3702.46) (32,3702.48) (64,3702.85) };
    \addlegendentry{algorithm=LPF-pruning-1}

    \addplot coordinates { (1,0.0) (2,0.0) (4,1.21686) (8,1.4882) (16,3.94044) (32,8.4479) (64,14.2843) };
    \addlegendentry{algo=pigz}

    \addplot coordinates { (1,59.7871) (2,118.603) (4,238.251) (8,472.626) (16,475.17) (32,476.606) (64,473.584) };
    \addlegendentry{algo=pxz}

    \addplot coordinates { (1,321.988) (2,335.578) (4,338.779) (8,343.703) (16,345.236) (32,344.29) (64,343.091) };
    \addlegendentry{algo=rec\_shard}
    \legend{}

    \nextgroupplot[
    title={einstein.de},
    ticks=none,
    ]
    \addplot coordinates { (1,5337.94) (2,5602.41) (4,5602.71) (8,5260.07) (16,5204.99) (32,5205.17) (64,5198.35) };
    \addlegendentry{algo=plz77}
    
    \addplot coordinates { (1,3704.7) (2,3704.87) (4,3704.97) (8,3704.99) (16,3705.12) (32,3705.49) (64,3706.06) };
    \addlegendentry{algorithm=LPF-pruning-1}

    \addplot coordinates { (1,0.0) (2,2.83856) (4,2.83856) (8,3.44335) (16,10.4034) (32,17.3752) (64,28.7154) };
    \addlegendentry{algo=pigz}

    \addplot coordinates { (1,130.914) (2,260.813) (4,512.242) (8,513.62) (16,512.229) (32,514.343) (64,514.141) };
    \addlegendentry{algo=pxz}

    \addplot coordinates { (1,262.168) (2,260.12) (4,261.016) (8,250.528) (16,246.369) (32,245.971) (64,243.512) };
    \addlegendentry{algo=rec\_shard}
    \legend{}

    \nextgroupplot[
    title={einstein.en},
    ylabel={space-overhead (\%)},
    xmajorticks=false,
    ]
    \addplot coordinates { (1,5201.34) (2,5200.95) (4,5200.98) (8,5511.22) (16,5400.89) (32,5267.27) (64,5267.73) };
    \addlegendentry{algo=plz77}
    
    \addplot coordinates { (1,3700.95) (2,3700.95) (4,3701.04) (8,3701.05) (16,3701.09) (32,3701.17) (64,3701.29) };
    \addlegendentry{algorithm=LPF-pruning-1}

    \addplot coordinates { (1,0.0) (2,0.0) (4,0.563056) (8,0.59385) (16,1.62288) (32,3.44826) (64,5.30081) };
    \addlegendentry{algo=pigz}

    \addplot coordinates { (1,26.0092) (2,52.9452) (4,104.347) (8,205.232) (16,411.435) (32,486.303) (64,486.34) };
    \addlegendentry{algo=pxz}
    
    \addplot coordinates { (1,285.318) (2,287.985) (4,283.86) (8,266.697) (16,272.578) (32,270.34) (64,263.121) };
    \addlegendentry{algo=rec\_shard}
    \legend{}

    \nextgroupplot[
    title={escherichia},
    ticks=none,
    ]
    \addplot coordinates { (1,5337.1) (2,5203.88) (4,5258.8) (8,5220.28) (16,5203.74) (32,5204.42) (64,5270.88) };
    \addlegendentry{algo=plz77}
    
    \addplot coordinates { (1,3703.84) (2,3704.04) (4,3704.16) (8,3704.42) (16,3704.49) (32,3704.69) (64,3704.95) };
    \addlegendentry{algorithm=LPF-pruning-1}

    \addplot coordinates { (1,0.0) (2,0.0) (4,2.46429) (8,4.20092) (16,7.66176) (32,14.5967) (64,32.0544) };
    \addlegendentry{algo=pigz}

    \addplot coordinates { (1,110.907) (2,190.752) (4,379.022) (8,457.646) (16,455.112) (32,467.354) (64,459.935) };
    \addlegendentry{algo=pxz}

    \addplot coordinates { (1,624.485) (2,618.517) (4,658.303) (8,666.543) (16,669.72) (32,669.476) (64,664.028) };
    \addlegendentry{algo=rec\_shard}
    \legend{}

    \nextgroupplot[
    title={influenza},
    ticks=none,
    ]
    \addplot coordinates { (1,5335.98) (2,5424.43) (4,5279.98) (8,5484.16) (16,5202.98) (32,5203.31) (64,5268.93) };
    \addlegendentry{algo=plz77}
    
    \addplot coordinates { (1,3702.87) (2,3702.85) (4,3703.07) (8,3703.11) (16,3703.07) (32,3703.29) (64,3703.74) };
    \addlegendentry{algorithm=LPF-pruning-1}

    \addplot coordinates { (1,0.0) (2,0.0) (4,1.70081) (8,2.24342) (16,5.49001) (32,12.073) (64,20.5867) };
    \addlegendentry{algo=pigz}

    \addplot coordinates { (1,78.2902) (2,158.809) (4,275.078) (8,419.255) (16,415.721) (32,415.988) (64,417.514) };
    \addlegendentry{algo=pxz}
    
    \addplot coordinates { (1,534.194) (2,549.815) (4,523.894) (8,596.71) (16,601.11) (32,602.396) (64,598.996) };
    \addlegendentry{algo=rec\_shard}
    \legend{}

    \nextgroupplot[
    title={kernel},
    xlabel={number of PEs},
    ylabel={space-overhead (\%)}
    ]
    \addplot coordinates { (1,5246.12) (2,5334.72) (4,5334.7) (8,5459) (16,5201.59) (32,5201.25) (64,5202.85) };
    \addlegendentry{algo=plz77}
    
    \addplot coordinates { (1,3701.68) (2,3701.71) (4,3701.83) (8,3701.84) (16,3701.85) (32,3702.04) (64,3702.19) };
    \addlegendentry{algorithm=LPF-pruning-1}

    \addplot coordinates { (1,0.0) (2,0.968361) (4,1.02225) (8,2.15769) (16,4.16147) (32,5.19845) (64,7.80077) };
    \addlegendentry{algo=pigz}

    \addplot coordinates { (1,47.7172) (2,95.9383) (4,189.041) (8,375.213) (16,471.532) (32,473.55) (64,473.965) };
    \addlegendentry{algo=pxz}

    \addplot coordinates { (1,258.45) (2,260.85) (4,259.704) (8,252.994) (16,252.692) (32,245.577) (64,240.183) };
    \addlegendentry{algo=rec\_shard}
    \legend{}

    \nextgroupplot[
    title={para},
    xlabel={number of PEs},
    ymajorticks=false,
    ]
    \addplot coordinates { (1,5201.11) (2,5467.37) (4,5440.34) (8,5228.33) (16,5467.83) (32,5732.74) (64,5732.87) };
    \addlegendentry{algo=plz77}
    
    \addplot coordinates { (1,3701.02) (2,3701.04) (4,3701.09) (8,3701.11) (16,3701.17) (32,3701.17) (64,3701.36) };
    \addlegendentry{algorithm=LPF-pruning-1}

    \addplot coordinates { (1,0.0) (2,0.580992) (4,0.679299) (8,1.29663) (16,1.94728) (32,3.92927) (64,7.92726) };
    \addlegendentry{algo=pigz}

    \addplot coordinates { (1,28.5252) (2,57.2701) (4,112.818) (8,196.701) (16,396.093) (32,419.973) (64,420.031) };
    \addlegendentry{algo=pxz}

    \addplot coordinates { (1,319.682) (2,327.193) (4,347.883) (8,340.873) (16,372.267) (32,374.069) (64,372.474) };
    \addlegendentry{algo=rec\_shard}
    \legend{}

    \nextgroupplot[
    title={world\_leaders},
    xlabel={number of PEs},
    ymajorticks=false,
    ]
    \addplot coordinates { (1,5342.7) (2,5296.5) (4,5296.56) (8,5520.56) (16,5274.91) (32,5279.14) (64,5274.57) };
    \addlegendentry{algo=plz77}
    
    \addplot coordinates { (1,3709.13) (2,3709.71) (4,3709.95) (8,3710.18) (16,3710.28) (32,3711.07) (64,3712) };
    \addlegendentry{algorithm=LPF-pruning-1}

    \addplot coordinates { (1,0.0) (2,0.0) (4,6.19994) (8,6.5044) (16,17.2393) (32,28.6023) (64,46.6805) };
    \addlegendentry{algo=pigz}

    \addplot coordinates { (1,265.016) (2,510.297) (4,517.352) (8,514.883) (16,516.418) (32,517.665) (64,515.713) };
    \addlegendentry{algo=pxz}
    
    \addplot coordinates { (1,299.034) (2,400.788) (4,412.371) (8,360.195) (16,425.573) (32,423.612) (64,429.314) };
    \addlegendentry{algo=rec\_shard}
    \legend{}
    
  \end{groupplot}
\end{tikzpicture}

\vspace{.5cm}
\begin{centering}
  \ref*{leg:throughput}
\end{centering}